\documentclass{article}

\usepackage[preprint]{neurips_2025}

\usepackage[utf8]{inputenc} 
\usepackage[T1]{fontenc}    
\usepackage{hyperref}       
\usepackage{url}            
\usepackage{booktabs}       
\usepackage{amsfonts}       
\usepackage{nicefrac}       
\usepackage{microtype}      
\usepackage{xcolor}         

\usepackage{amssymb}
\usepackage{amsmath}
\usepackage{amsthm}
\usepackage{graphicx}
\usepackage{subfigure}
\usepackage{hyperref}
\usepackage[linesnumbered,boxed,noend]{algorithm2e}

\newcommand{\be}{\begin{equation}}
\newcommand{\ee}{\end{equation}}

\newcommand{\sset}[1]{\left\{#1\right\}}
\newcommand{\prob}[1]{\hbox{Pr}\left\{#1\right\}}

\newtheorem{thm}{Theorem}
\newtheorem{cor}{Corollary}
\newtheorem{lem}{Lemma}

\title{Steady Continuous Monitoring is (Just Barely) Impossible for Tests of Unbounded Length}

\author{Eric Bax and Alex Shtoff}

\begin{document}

\maketitle

\begin{abstract}
AB testing evaluates the difference between a control and a treatment in a statistically rigorous manner. Continuous monitoring allows statistical evaluation of an AB test as it proceeds. One goal of continuous monitoring is early stopping -- confirming a statistically significant difference between control and treatment as soon as possible. Another goal is to maintain some statistical capability to discover significant differences later in the test if they cannot be confirmed earlier. These goals are in conflict -- looser requirements for early stopping leave us with more stringent ones for later. This paper shows that it is impossible to maintain a constant requirement for significance for tests that have no a priori stopping time, but we can come arbitrarily close to that goal by using tests that require repeated significant results to confirm statistically significant differences between treatment and control.
\end{abstract}



\section{Introduction}
Recently, there has been renewed interest in \textit{group sequential methods} \cite{wald47,peto76,pocock77,wang87,jennison00,lewis23}, or \textit{interim analysis methods} \cite{huang17,grayling18,ciolino23} for long-running AB tests, which may lack a priori limits on run time or number of observations. We refer to such tests as unbounded AB tests. Methods that control Type I error for such tests include \textit{always-valid bounds} and \textit{continuous monitoring methods} \cite{deng16,johari21,maharaj23,grunwald24,waudbysmith24}. 

Most group sequential methods use some form of $\alpha$ spending -- partitioning a Type I error budget over the interim analyses \cite{lan83,lewis23}. Adapting them to unbounded tests requires partitioning the error budget over an unbounded number of interim analyses. Dividing the budget more finely (while controlling Type I error) makes conditions for ending the test more stringent, increasing test time and expense and increasing the likelihood of Type II error. 

To counter those effects, always-valid bounds and continuous monitoring methods often rely on correlation between $p$-values over sequences of interim analyses. Those correlations follow from a combination of properties of specific statistics and requirements for the observations. For example, a running average over observations combined with a limit on the value of each observation implies that the running average will have limited change as the number of observations grows large. 

Another method to avoid extremely stringent conditions for large numbers of interim analyses is to require repeated statistically significant results \cite{bax_sarkar_shtoff} over multiple interim analyses to declare a statistically significant result for the AB test as a whole. Requiring repetition makes each interim significance requirement less stringent by multiplying the $p$-value required for significance by the number of repeats required -- a consequence of nearly uniform validation \cite{bax_compression,bax16}. 

Compared to other methods, requiring repetition removes the need to know a priori, prove, or empirically evaluate whether side conditions on the data hold, because it does not rely on correlations among $p$-values over sequences of interim analyses. Instead, it allows (and requires) the data to prove itself. This makes it simple and widely applicable, yet effective, as has been shown for bounded tests \cite{bax_sarkar_shtoff}, and as we show in this paper for unbounded tests. 

Section \ref{prelim_sec} reviews the concepts of $\alpha$ spending for sequential testing, nearly uniform validation, and requiring repeated significance. Section \ref{unbounded_sec} shows how to apply those concepts to unbounded tests. Section \ref{geo_sec} explores using $\alpha$ spending based on geometric series to produce test plans that control Type I error. Section \ref{imp_sec} shows that, for unbounded tests, the $p$-values required for statistically significant results must grow eventually, but we can delay that using $\alpha$ spending based on $p$-series with exponents near one. Section \ref{disc_sec} concludes with a discussion of possibilities for future work.

\section{Preliminaries} \label{prelim_sec}
For a statistical test, a null hypothesis is an assertion about the (unknown) out-of-sample statistics or underlying distribution that generates the observations in the test. For a pre-specified metric measured on the in-sample observations, a $p$-value is the probability that the metric takes a value at least as extreme as the measured value under the null hypothesis. 

A Type I error occurs when a statistical test rejects the null hypothesis but the null hypothesis holds -- a false positive. In test planning, we specify an acceptable upper bound for the probability of Type I error. Call that bound $\alpha$. If the test rejects the null hypothesis, then we say it is rejected at confidence level $1 - \alpha$. 

A simple test plan specifies a metric and $\alpha$. After all in-sample observations are collected, the test measures the metric and produces a $p$-value based on the measured value. If $p \leq \alpha$, then the null hypothesis is rejected. 

For example, in an AB test, the null hypothesis may be that both treatments A and B (one of which may be the control) have the same average for some metric. If the $p$-value for the test is below a specified $\alpha$, then we reject the null hypothesis, asserting that the difference in treatments has a statistically significant impact on the metric, with confidence $1 - \alpha$. 

For many AB tests, it is useful to allow early stopping -- ending the test with a statistically significant result before collecting all planned observations. In tests of user interfaces or advertisements on apps and web pages, this can reduce time to market and save money. In medical tests, it can save lives by delivering useful treatments to the control group as well as the treatment group as soon as there is a positive result. 

Early stopping can create the risk of data dredging or $p$-hacking \cite{smith02,wasserstein02,young11} in time, by stopping when a metric is positive but continuing when it is negative. That can introduce large gaps between in-sample and out-of-sample results, due to a selection effect \cite{bax_compare_prices,pocock05,wang16,huang17}. 

In general, stopping when the $p$-value is less than $\alpha$ for the data collected from the start of the test to an interim analysis fails to bound the Type I error probability by $\alpha$. If the probability of Type I error at each of $d$ interim analyses is at most $\alpha$, then the probability of Type I error at any interim analysis may be as high as $d \alpha$. (In the worst case, each interim Type I error probability is $\alpha$, and those events are disjoint.)

A technique called $\alpha$ spending avoids this problem by partitioning $\alpha$ into $\alpha_1, \ldots, \alpha_d$ that sum to $\alpha$ and allowing early stopping when the $p$-value for interim analysis $k$ (call it $p_k$) has $p_k \leq \alpha_k$. In this case, the probability of Type I error at each interim analysis $k$ is at most $\alpha_k$, so the overall probability of a Type I error is at most $\alpha$, by the sum bound on the union of probabilities, also known as Boole's inequality \cite{boole47,casella01}. (Dividing evenly: $\alpha_k = \frac{\alpha}{d}$ is known as a Bonferroni correction \cite{bonferroni36}.) The method of $\alpha$ spending has a long history in group sequential testing \cite{wald47,haybittle71,peto76,pocock77,wang87,jennison00,lewis23}, especially for medical trials. 

As $\alpha$ is partitioned more finely over more interim analyses, the $p$-value requirements for statistical significance become more stringent. Requiring repeated significance in order to stop the test counteracts this effect, because it allows less stringent requirements for statistical significance while maintaining the $\alpha$ bound on probability of Type I error. If we require $p_k \leq \delta_k$ for at least $r$ interim analyses before rejecting the null hypothesis and ending the test, then the probability of Type I error is at most $\frac{\delta_1 + \ldots + \delta_d}{r}$, because the probability of Type I error is then maximized by having $p_k \leq \delta_k$ co-occur for exactly $r$ different $k$ values in $1, \ldots, d$ when it occurs for any $k$. 

Requiring $r$ failures for a Type I error allows us to divide the sum bound by $r$ (see Figure \ref{fig_nearly_uniform}), or, equivalently to multiply each budgeted $\alpha_k$ by $r$ to get $p$-value requirements for interim analyses, since $\alpha_1 + \ldots + \alpha_d = \alpha$ implies that if we set $p$-value requirements $\delta_k = \alpha_k r$, then $\frac{\delta_1 + \ldots + \delta_d}{r} = \alpha$. 

\begin{figure}
\centering
\includegraphics[width=3.5in]{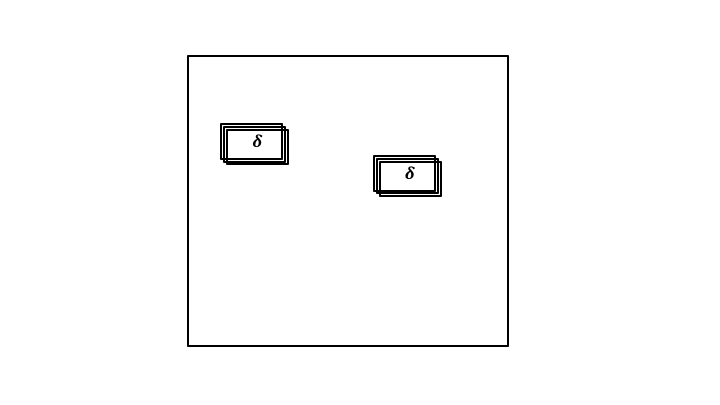} 
\caption{\textbf{Requiring Repetition}. For $r = 3$ required repeats over $d = 6$ interim analyses each with interim analysis Type I error probability $\delta$, the worst-case probability of Type I error for the entire test is $\frac{6 \delta}{3} = 2 \delta$, since 3 interim analyses must all have Type I errors for the entire test to have it. (If six carpets each cover $\delta$ area, then at most $2 \delta$ area can be covered at least three deep.) (Figure from \cite{bax16}.)} \label{fig_nearly_uniform}
\end{figure}

An earlier paper \cite{bax_sarkar_shtoff} applies this insight to develop repetition-based test plans for tests with a priori limits on the number $d$ of interim analyses. For that case, it is possible, for example, to spend $\alpha$ evenly over interim analyses: $\alpha_k = \frac{\alpha}{d}$ for all $k$ in $1, \ldots, d$. For unbounded $d$, that is impossible. This paper focuses on methods to spend $\alpha$ and select the number of required repeats for unbounded tests. 

\section{From Bounded to Unbounded Tests} \label{unbounded_sec}
Following the notation from \cite{bax_sarkar_shtoff}, we will refer to interim analyses as decision points. For each decision point $t$, let $\delta_t$ be the upper bound for the $p$-value at decision point $t$ to be considered significant, and let $r_t$ be the number of significant $p$-values over decision points $1$ to $t$ required to end the test and reject the null hypothesis. 

As an example, suppose that for each decision point $t$, if at least half of the decision points $k$ from $1$ to $t$ have $p_k \leq \delta_k$, where $p_k$ is the $p$-value at decision point $k$, then we end the test and declare a significant result -- rejecting the null hypothesis. Then $r_t = \lceil \frac{t}{2} \rceil$ for all $t \geq 1$.

\begin{lem} \label{lp_lemma}
Consider a test, specified by an integer $d > 0$ and nonnegative values $\delta_1, \ldots, \delta_d$ and $r_1 \leq \ldots \leq r_d$, that stops and rejects the null hypothesis at the first decision point $t$ for which at least $r_t$ of the decision points $k$ from $1$ to $t$ each have $p_k \leq \delta_k$, with $p_k$ the $p$-value for the null hypothesis based on the data up to decision point $k$. 
\begin{itemize}
\item Let $F_k$ be the event that, at decision point $k$, $p_k \leq \delta_k$, but the null hypothesis holds. (Refer to this event as false significance at decision point $k$.)
\item Let $A(S)$ be the event that $F_k$ occurs for all $k \in S$ and does not occur for any $k \not\in S$:
\be
\forall S \subseteq \sset{1, \ldots, d}: \hbox{ } A(S) = \left(\bigcap_{k \in S} F_k\right) \cap \left(\bigcap_{k \not\in S} \overline{F_k}\right).
\ee
(These are the atoms of the joint probability space over events $F_k$.)
\item For $S \subseteq \sset{1, \ldots. d}$, let $p_S$ be the probability of $A(S)$:
\be
p_S \equiv \prob{A(S)}.
\ee
\item Let $Q$ be the set of sets $S$ such that false significance at all decision points indicated by $S$ (and at no others) results in a Type I error. This occurs if there is some $t \leq d$ with at least $r_t$ elements from $\sset{1, \ldots, t}$ in $S$:
\be
Q \equiv \sset{S | \exists t \in \sset{1, \ldots, d}: \left|\sset{1, \ldots, t} \cap S\right| \geq r_t}.
\ee
\end{itemize}
Then any solution to the following linear program is a joint probability distribution of events $F_k$ that maximizes Type I error:
\be
\max_{\left(p_S | S \subseteq \sset{1, ..., d}\right)} \sum_{S \in Q} p_S
\ee
given constraints:
\begin{enumerate}
\item The probabilities of the atoms of the joint distribution sum to one:
\be
\sum_{S \subseteq \sset{1, \ldots, d}} p_S = 1.
\ee
\item The probabilities are nonnegative:
\be
\forall S \subseteq \sset{1, \ldots, d}: \hbox{ } p_S \geq 0.
\ee
\item Probability of false significance at decision point $k$ ($p_k \leq \delta_k$ and null hypothesis holds) is at most $\delta_k$:
\be
\forall k \in \sset{1, \ldots, d}: \hbox{ } \sum_{S \supseteq \sset{k}} p_S \leq \delta_k.
\ee
\end{enumerate}
\end{lem}

\begin{proof}
By the definition of $Q$, the maximized sum is over the set of all probability space atoms $A(S)$ that produce Type I errors. Since each probability atom's event is a fully-specified joint outcome for false significances $F_1$ to $F_k$, the atoms correspond to mutually exclusive events in the joint probability space over $F_1$ to $F_k$. So there is no need to subtract probabilities of intersections for those events -- they do not intersect each other. Hence their sum is the probability of Type I error given a joint distribution specified by $p_S$ values. 

The first two constraints enforce that the $p_S$ values in the solution of the linear program correspond to a probability distribution -- probabilities are nonnegative and sum to one. The last constraint enforces a bound of $\delta_k$ on the probability of each $F_k$. (Since we require a $p$-value $p_k \leq \delta_k$ for significance at each decision point $k$, the probability of significance at $k$ if the null hypothesis holds, which is the probability of $F_k$, is at most $\delta_k$.)
\end{proof}

Bounding the result from the linear program in Lemma \ref{lp_lemma} produces a simpler bound on Type I error:

\begin{thm} \label{delta_thm}
Given nonnegative $\delta_1, \delta_2, \ldots, \delta_d$ and $r_1 \leq r_2 \leq \ldots \leq r_d$, consider a test that stops at the first decision point $t$ at which at least $r_t$ of the decision points $k$ from $1$ to $t$ each have $p_k \leq \delta_k$, with $p_k$ the $p$-value at decision point $k$ based on data from the start of the test until decision point $k$ and rejects the null hypothesis, or otherwise runs through decision point $d$ without rejecting the null hypothesis. An upper bound on the Type I error probability for that test is
\be
\sum_{t} \frac{\delta_t}{r_t}
\ee
\end{thm}

\begin{proof}
We want to show that 
\be
\sum_{S \in Q} p_S^* \leq \sum_{t} \frac{\delta_t}{r_t},
\ee
where $\left(p_S^* | S \subseteq \sset{1, \ldots, d}\right)$ is an optimal solution for the linear program of Lemma \ref{lp_lemma}, which maximizes Type I error probability given that the null hypothesis holds. 

Recall that
\be
Q = \sset{S | \exists t \in \sset{1, \ldots, d}: \left|\sset{1, \ldots, t} \cap S\right| \geq r_t}.
\ee
For each $S \in Q$, let $t(S)$ be the decision point at which the test ends, given joint event $A(S)$:
\be
t(S) \equiv \min \sset{t | S \cap \sset{1, \ldots, t}| \geq r_t}.
\ee
Also, let
\be
B(S) \equiv S \cap \sset{1, \ldots, t(S)}.
\ee
Note that $|B(S)| \geq r_{t(S)}$ because at least $r_{t(S)}$ significant results ($p_k \leq \delta_k$) are required to end the test at decision point $t(S)$.

Since $|B(S)| \geq r_{t(S)}$, 
\be
\sum_{S \in Q} p_S^* 
\ee
\be
\leq \sum_{S \in Q} \sum_{k \in B(S)} \frac{p_S^*}{r_{t(S)}}.
\ee
For all $k \in B(S)$, $k \leq t(S)$, because $B(S) \subseteq \sset{1, \ldots, t(S)}$. Also we assume $r_1 \leq \ldots \leq r_{t(S)}$, so $\forall k \in B(S): r_k \leq r_{t(S)}$. So the previous sum is
\be
\leq \sum_{S \in Q} \sum_{k \in B(S)} \frac{p_S^*}{r_k}.
\ee
Substitute $t$ for $k$ and reorder the sums:
\be
= \sum_{t \in \sset{1, \ldots, d}} \sum_{S \in Q \land t \in B(S)} \frac{p_S^*}{r_t}.
\ee
Remove $\frac{1}{r_t}$ from the last sum:
\be
= \sum_{t \in \sset{1, \ldots, d}} \frac{1}{r_t} \sum_{S \in Q | t \in B(S)} p_S^*.
\ee
Expand the last sum to be over all $S \supseteq \sset{t}$:
\be
\leq \sum_{t \in \sset{1, \ldots, d}} \frac{1}{r_t} \sum_{S \supseteq \sset{t}} p_S^*.
\ee
By the last constraint of the linear program in Lemma \ref{lp_lemma}, the last sum is at most $\delta_t$:
\be
\leq \sum_{t \in \sset{1, \ldots, d}} \frac{\delta_t}{r_t}.
\ee
\end{proof}

For an unbounded test, we do not know the number of decision points a priori. So to bound the Type I error probability, we must bound the bounds from Theorem \ref{delta_thm} over values $d$ from $1$ to $\infty$:

\begin{thm} \label{delta_thm_inf}
For an unbounded test, specified by infinite sequences of nonnegative $\delta_1, \delta_2, \ldots$ and of positive integers $r_1 \leq r_2 \leq \ldots$, that stops at the first decision point $t$ at which at least $r_t$ of the decision points $k$ from $1$ to $t$ each have $p_k \leq \delta_k$, with $p_k$ the $p$-value at decision point $k$ based on data from the start of the test until decision point $k$ and declares the null hypothesis false, or otherwise runs forever without declaring the null hypothesis to be false, an upper bound on the Type I error probability for that test is
\be
\sum_{t = 1}^{\infty} \frac{\delta_t}{r_t}.
\ee
\end{thm}

\begin{proof}
Consider two cases: (1) $d$ is infinite, and (2) $d$ is finite, but could be any finite positive integer. In case (1) the test runs forever, therefore it never stops and hence never rejects the null hypothesis. But rejecting the null hypothesis is a necessary condition for a Type I error, so there is no Type I error. 

To handle case (2) we must find an upper bound for the bound from Theorem \ref{delta_thm} over all positive integer values for $d$:
\be
\max_{0 < d < \infty} \sum_{t = 1}^{d} \frac{\delta_t}{r_t}.
\ee
Since $\forall t: \delta_t \geq 0$, the sums for successive values of $d$ are a monotonically increasing series. So its maximum is
\be
= \lim_{d \rightarrow \infty} \sum_{t = 1}^{d} \frac{\delta_t}{r_t}.
\ee
And this is the definition of the infinite series:
\be
\sum_{t = 1}^{\infty} \frac{\delta_t}{r_t}.
\ee
\end{proof}

We can restate Theorem \ref{delta_thm_inf} from an $\alpha$-spending point of view:

\begin{cor} \label{alpha_cor}
Given $r_1 \leq r_2 \leq \ldots$ and a set of nonnegative values $\alpha_t$ such that
\be
\sum_{t = 1}^{\infty} \alpha_t = \alpha
\ee
if a test stops at the first decision point $t$ at which at least $r_t$ of the decision points up to $t$ have $p$-values with:
\be
p \leq \alpha_t r_t,
\ee
then the test has Type I error probability at most $\alpha$.
\end{cor}

\begin{proof}
Let $\delta_t = \alpha_t r_t$ and apply Theorem \ref{delta_thm_inf}.
\end{proof}

\section{Geometric $\alpha$-Spending for Unbounded Tests} \label{geo_sec}
For bounded tests -- ones that have an a priori limit $d$ on the number of decision points -- it is possible to use uniform $\alpha$ spending, with all $\alpha_k = \frac{\alpha}{d}$, where $d$ is the limit on the number of decision points. Similarly, it is possible to budget $\alpha$ so that the $p$-value requirement for significance, $\delta_t = \alpha_t r_t$, is equal for all $t \in \sset{1, \ldots, d}$. Unbounded tests require other strategies.

One strategy for unbounded tests, \textit{geometric $\alpha$-spending} \cite{bax_sarkar_shtoff}, allocates the same fraction of the remaining $\alpha$ budget to each successive decision point: select a ``withdrawal rate" $0 < w < 1$ and set $\alpha_1 = w \alpha$, $\alpha_2 = w (1 - w) \alpha$, and in general $\alpha_t = w (1 - w)^{t - 1} \alpha$. Decreasing $w$ decreases early $\alpha_t$ values, saving more of $\alpha$ for later ones. Geometric spending works for unbounded tests, since the sum of $\alpha_t$ values is the sum of a geometric series that sums to $\alpha$ in the limit. 

Applying Corollary \ref{alpha_cor} to geometric $\alpha$ spending, if the test plan is to require at least $r_t = \lceil u t \rceil$ significant results over decision points 1 to $t$ to stop at decision point $t$, then significance for each decision point $t$ requires
\be
p_t \leq w (1 - w)^{t - 1} \alpha \lceil u t \rceil \label{geo_formula}
\ee
for a value $0 < w < 1$ specified as part of the test plan. 

To find the maximum required $p_t$, take the derivative of the RHS of Inequality \ref{geo_formula} with respect to $t$, set it to zero, and solve for $t$: 
\be
t_{\min} \approx \frac{-1}{\ln (1 - w)}. \label{geo_tmin}
\ee
(We have $\approx$ rather than $=$ only because we use $ut$ in place of $\lceil u t \rceil$ to find $t_{\min}$.) 

Conversely, solve this for $w$ to find a value that makes the least stringent $p_t$ requirement occur at decision point $t$:
\be
w \approx 1 - e^{-1/t}. \label{geo_wmin}
\ee
This can be useful if we have some idea how long the test is likely to run before achieving statistically significant results. Equations \ref{geo_tmin} and \ref{geo_wmin} show that decreasing the ``withdrawal" rate $w$ increases the number of decision points $t$ before reaching the maximum (least stringent) required $p_t$. 

Figure \ref{fig_zrhovw} compares using geometric $\alpha$ spending with repetition requirement $u = 0.2$ to an always-valid method based on autocorrelation among averages \cite{waudbysmith24} (their Equation 8, page 8), with probability of Type I error at most $\alpha = 0.05$. (We plot $Z$ scores rather than $p$ values for ease of presentation, since $Z$ scores are familiar to many practitioners and have a 1-1 correspondence with $p$ values. The plotted $Z$ scores are the number of standard deviations from the mean of a standard normal such that the left and right tails together have probability $p$. For comparison, a test with a single decision point that spends all of $\alpha = 0.05$ on that decision point would require $Z \approx 1.96$ for $p = 0.05$.) The $Z$ score required for significance for that method is
\be
\sqrt{\frac{2(t \rho^2 + 1)}{t \rho^2} \ln \left(\frac{\sqrt{t \rho^2 + 1}}{\alpha}\right)},
\ee
with $\rho$ a parameter that, like $w$ for our method, controls the $t$ value at which the $Z$ score reaches a minimum. Both methods apply to unbounded tests.

For $u = 0.2$, repetition generally produces lower $Z$ score requirements. However, requiring repetition with $u = 0.2$ may require a test to run for up to 20\% longer, if it remains significant once it becomes significant for some $t$ value. Of course, due to the less stringent $Z$ score requirement, requiring repetition may also result in earlier stopping than for the other always-valid method. For $u \approx 0.15$, both methods have similar $Z$ score minimums over $t$. 

\begin{figure}
\centering
\includegraphics[width=3.5in]{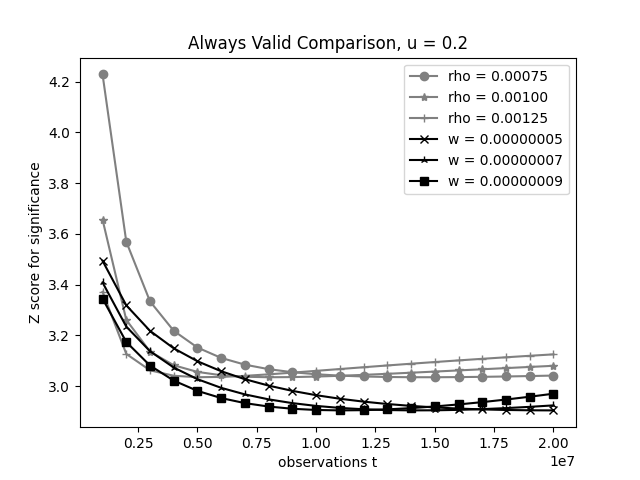}
\caption{\textbf{Comparison to another always-valid method.} For $u = 0.2$, requiring repetition gives less stringent requirements for significance than a method based on autocorrelation for averages  \cite{waudbysmith24}.} \label{fig_zrhovw}
\end{figure}

\section{Nearly Impossible is Possible} \label{imp_sec}
Figure \ref{fig_zrhovw} has $Z$ score curves that increase after reaching a minimum, prompting the question: ``Is there a way to spend $\alpha$ that gives a flat curve for an unbounded test, keeping the required $p$-value from shrinking to zero, or, equivalently, the required $Z$ score from becoming unbounded as $t \rightarrow \infty$?"

To answer that question, first consider that, in general, for an unbounded test we require 
\be
\sum_{t = 1}^{\infty} \alpha_t \leq \alpha
\ee
to avoid over-spending $\alpha$. If we let 
\be
\alpha_t = x_t \alpha,
\ee
then we require 
\be
\sum_{t = 1}^{\infty} x_t \leq 1,
\ee
so we need $x_t$ to be a convergent series. 

Now we show that the answer to the question is a negative of the form: ``It is not quite possible, but nearly so."

\begin{thm} \label{impossible_thm}
It is impossible to have all $p$-value requirements $\delta_t$ given by Corollary \ref{alpha_cor} equal for an unbounded test if $r_t = \lceil u t \rceil$ for a constant $u > 0$.
\end{thm}

\begin{proof}
Recall that
\be
\delta_t = \alpha_t r_t = x_t \alpha \lceil u t \rceil.
\ee
Suppose these values are all equal to a constant $c$. Then, for all $t \in \sset{1, 2, \ldots}$, 
\be
x_t = \frac{c}{\alpha \lceil u t \rceil} > \frac{c}{\alpha (u t + 1)}
\ee
\be
= \frac{c}{\alpha u} \left(\frac{1}{t + \frac{1}{u}}\right).
\ee
We require the sum $x_1 + x_2 + \dots \leq 1$, but the sum of these terms diverges. Informally, it diverges because the terms are on the order of harmonic series terms: $\frac{1}{t}$, and the harmonic series diverges.

More formally, note that:
\be
\frac{1}{t + \frac{1}{u}} = \frac{1}{t} - \frac{\frac{1}{u}}{t (t + \frac{1}{u})}.
\ee
So
\be
\sum_{t = 1}^{\infty} x_t > 
\ee
\be
\frac{c}{\alpha u} \left[\sum_{t = 1}^{\infty} \frac{1}{t} - \frac{1}{u} \sum_{t = 1}^{\infty} \frac{1}{t (t + \frac{1}{u})} \right] \label{sums}
\ee
The second sum is less than
\be
\sum_{t = 1}^{\infty} \frac{1}{t^2}.
\ee
We know that the $p$-series 
\be
\sum_{t = 1}^{\infty} \frac{1}{t^v}
\ee
diverges for $0 \leq v \leq 1$ and converges for $v > 1$. So in Expression \ref{sums}, the first sum diverges and the second converges, implying that the expression diverges.
\end{proof}

In short, we cannot have $x_t$ proportional to $\frac{1}{t}$ because its $p$-series diverges. However, for any $v > 1$, we can use $x_t$ proportional to $\frac{1}{t^v}$, since its $p$-series converges. So, in a sense, Theorem \ref{impossible_thm} is the most nearly possible impossible result possible. 

Since 
\be
\sum_{t = 1}^{\infty} \frac{1}{t^v} = \zeta(v),
\ee
where $\zeta()$ is the Riemann zeta function, we can set 
\be
\forall t: \hbox{ } x_t = \frac{1}{\zeta(v) t^v}
\ee
for any $v > 1$. As $v \rightarrow 1$, we flatten the $\delta_t$ curve, but $\zeta(v)$ increases, making the initial $p$-value requirement more stringent. For example, $\zeta(2) = \frac{\pi^2}{6}$, so we can set
\be
x_t = \frac{6}{\pi^2 t^2},
\ee
and get
\be
\delta_t = \frac{6}{\pi^2 t^2} \alpha \lceil u t \rceil \geq \frac{6 \alpha u}{\pi^2 t}.
\ee

To smooth the curve of $\delta_t$ over $t$, it is possible to remove the head of the $p$-series and use only its tail for $x_t$ values. Remove the first $s$ terms and let 
\be
h(s, v) = \sum_{t = 1}^{s} \frac{1}{t^v} 
\ee
be the head of the $p$-series consisting of the first $s$ terms. Removing that head yields
\be
x_t = \frac{1}{\left[\zeta(v) - h(s, v)\right] \left(t + s\right)^v},
\ee
so
\be
\delta_t = \frac{1}{\left[\zeta(v) - h(s, v)\right] \left(t + s\right)^v} \alpha \lceil u t \rceil.
\ee

Figure \ref{fig_zpseries_tail} shows the results of removing some initial terms from the $p$-series to produce $x_t$. The curve for $s = 0$ has no terms removed. Note that removing more terms lowers the required $Z$ scores for significance for most of the $t$ values. 

\begin{figure} 
\centering
\includegraphics[width=3.5in]{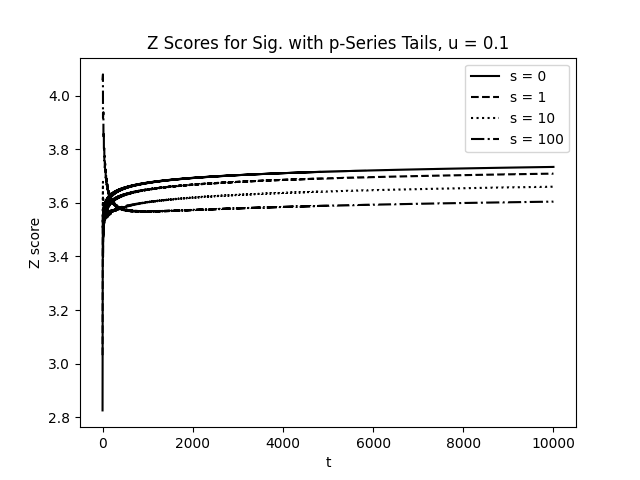}
\caption{\textbf{Headless $p$-series $\alpha$ spending: $\mathbf{Z}$ scores for significance.} For $p$-series $\alpha$ spending with $v = 1.1$, $\alpha = 0.05$, and $u = 0.1$, removing more initial values from the $p$-series before starting to allocate the series values to $x_t$ lowers required $Z$ scores on the right.} \label{fig_zpseries_tail}
\end{figure}

\section{Discussion} \label{disc_sec}
We have explored different ways to spend $\alpha$ with repetition requirements for unbounded tests, showing that it is possible to adjust the levels of significance required at each decision point by requiring different amounts of repetition. We have also shown that it is impossible to have a significance requirement that does not eventually become more stringent for an unbounded test with a repetition requirement that is a fixed fraction of the prior decision points, but we can stave that off for an arbitrary number of decision points by setting parameters appropriately. 

An earlier paper \cite{bax_sarkar_shtoff} that focused on bounded tests discussed how to balance control of Type I and Type II error, by spending a portion of $\alpha$ on a single loose requirement for significance at the end of the test, which is reminiscent of traditional sequential testing strategies, such as the Haybittle-Peto boundary \cite{haybittle71,peto76}. It is a challenge to imagine something similar for unbounded tests, since their ending decision points are not known a priori. Perhaps spending a portion of $\alpha$ for loose requirements at a (possibly unbounded) subsequence of the decision points, either with or without a repetition requirement, would prove useful. 

It would also be interesting to explore ways to combine repetition requirements with other approaches to continuous monitoring and always-valid bounds \cite{deng16,johari21,maharaj23,grunwald24,waudbysmith24}. In some cases, it may be possible to improve both types of bounds by combining their strategies. The challenge is that non-repetition methods sometimes make use of side conditions that are instead proven empirically in some sense with repetition requirements -- we need to analyze how these approaches overlap to understand how they can complement each other beyond mere duplication, if possible.

\bibliographystyle{unsrt}
\bibliography{bax}

\end{document}